\documentclass[journal]{IEEEtran}
\usepackage{graphicx,color}
\usepackage{cite}
\usepackage{setspace} 
\usepackage{amsmath}
\usepackage{amsmath,amsthm,amssymb}
\usepackage{multirow}
\usepackage{hhline}
\usepackage{epsfig}
\usepackage{subfigure}
\usepackage{epstopdf}
\usepackage{verbatim}
\usepackage{algorithm}
\usepackage{algorithmicx}
\usepackage{algpseudocode}
\usepackage{etoolbox}
\usepackage{cases}
\usepackage{csquotes}
\usepackage{mathtools,stmaryrd}
\SetSymbolFont{stmry}{bold}{U}{stmry}{m}{n}
\usepackage{bbm}
\usepackage{lettrine}

\newcommand{\suchthat}{\;\ifnum\currentgrouptype=16 \middle\fi|\;}

\makeatletter
\newcommand*{\indep}{%
  \mathbin{%
    \mathpalette{\@indep}{}%
  }%
}
\newcommand*{\nindep}{%
  \mathbin{
    \mathpalette{\@indep}{\not}
  }%
}
\newcommand*{\@indep}[2]{%
  \sbox0{$#1\perp\m@th$}
  \sbox2{$#1=$}
  \sbox4{$#1\vcenter{}$}
  \rlap{\copy0}
  \dimen@=\dimexpr\ht2-\ht4-.2pt\relax
  \kern\dimen@
  {#2}%
  \kern\dimen@
  \copy0 
} 
\makeatother

\makeatletter
\newcommand*{\algrule}[1][\algorithmicindent]{%
  \makebox[#1][l]{%
    \hspace*{.2em}
    \vrule height .75\baselineskip depth .25\baselineskip
  }
}

\newcount\ALG@printindent@tempcnta
\def\ALG@printindent{%
    \ifnum \theALG@nested>0
    \ifx\ALG@text\ALG@x@notext
    \else
    \unskip
    \ALG@printindent@tempcnta=1
    \loop
    \algrule[\csname ALG@ind@\the\ALG@printindent@tempcnta\endcsname]%
    \advance \ALG@printindent@tempcnta 1
    \ifnum \ALG@printindent@tempcnta<\numexpr\theALG@nested+1\relax
    \repeat
    \fi
    \fi
}
\patchcmd{\ALG@doentity}{\noindent\hskip\ALG@tlm}{\ALG@printindent}{}{\errmessage{failed to patch}}
\patchcmd{\ALG@doentity}{\item[]\nointerlineskip}{}{}{} 
\makeatother

\newtheorem{prop}{Proposition}

\begin{document}

\title{MIMO SWIPT Systems with Power Amplifier Nonlinearities and Memory Effects}
\author{Authors}
\author{Priyadarshi Mukherjee, \textit{Member, IEEE}, Souhir Lajnef, and~Ioannis Krikidis, \textit{Fellow, IEEE}
\thanks{P. Mukherjee, S. Lajnef, and I. Krikidis are with the Department of Electrical and Computer Engineering, University of Cyprus, Nicosia 1678 (E-mail: \{mukherjee.priyadarshi, lajnef.souhir, krikidis\}@ucy.ac.cy). This work was co-funded by the European Regional Development Fund and the Republic of Cyprus through the Research and Innovation Foundation, under the projects INFRASTRUCTURES/1216/0017 (IRIDA) and EXCELLENCE/1918/0377 (PRIME). It has also received funding from the European Research Council (ERC) under the European Union's Horizon 2020 research and innovation programme (Grant agreement No. 819819).}
}

\maketitle
\begin{abstract}
In this letter, we study the impact of nonlinear high power amplifier (HPA) on simultaneous wireless information and power transfer (SWIPT), for a point-to-point multiple-input multiple-output communication system. We derive the rate-energy (RE) region by taking into account the HPA nonlinearities and its associated memory effects. We show that HPA significantly degrades the achievable RE region, and a predistortion technique is investigated for compensation. The performance of the proposed predistortion scheme is evaluated in terms of RE region enhancement. Numerical results demonstrate that approximately $24\%$ improvement is obtained for both power-splitting and time-splitting SWIPT architectures.
\end{abstract}


\IEEEpeerreviewmaketitle

\section{Introduction}

\lettrine[lines=2]{S}{imultaneous} wireless information and power transfer (SWIPT) is a promising technology to provide energy sustainability and ubiquitous connectivity to low power wireless devices \cite{intro1,intro2}. As there will be huge deployment of smart devices in the era of massive machine-type communications and internet of things, it is practically impossible to regularly recharge these devices \cite{srv}. Hence, SWIPT has gained increased attention over the last few years, as it enables simultaneous communication and energy harvesting through fully-controlled dedicated radio-frequency (RF) signals. 

For proper implementation of SWIPT, we require energy harvesting at the receiver with higher efficiency. Various strategies and signal waveforms exist in literature \cite{intro2,plinear} that take into account the nonlinearity of the rectification circuit. Specifically, experimental results demonstrate that waveforms with high peak-to-average-power-ratio (PAPR) e.g., multi-sine, chaotic signals etc, boost the wireless energy transfer efficiency \cite{papr1,papr2}. However, high PAPR signals are more susceptible to high-power amplifier (HPA) nonlinearities.

A typical HPA is unable to successfully transmit unclipped high PAPR signals. As a result, the transmitted signal gets distorted even before its transmission. Moreover, a practical HPA also exhibits some degree of memory i.e., the HPA output not only depends on the current HPA input, but also on the past HPA inputs \cite{mmry}. Despite these experimentally demonstrated observations, existing works do not consider the effects of HPA on the rate-energy (RE) performance of SWIPT. Hence we investigate these effects in SWIPT systems.

In this letter, we model the characteristics of HPA and we study its effects by considering a basic point-to-point multiple-input-multiple-output (MIMO) SWIPT system. We prove that the distortions introduced by the HPA nonlinearities and its associated memory effects are not independent of the input signal; these HPA-based practical impairments are taken into account in our SWIPT analysis. We characterize the achievable RE region of the considered MIMO SWIPT system, and we show that HPA significantly squeezes SWIPT performance. As a result, a noteworthy enhancement of the RE region is attained by opting for predistortion, which linearizes the HPA performance up to its saturation; the proposed predistortion enables the transmission of unclipped high PAPR signals. To the best of authors' knowledge, this is the first work that deals with the impact of PAPR on MIMO SWIPT systems. 

\section{System Model} \label{msec}

\begin{figure*}[!t]
\centering\includegraphics[width=0.8\linewidth]{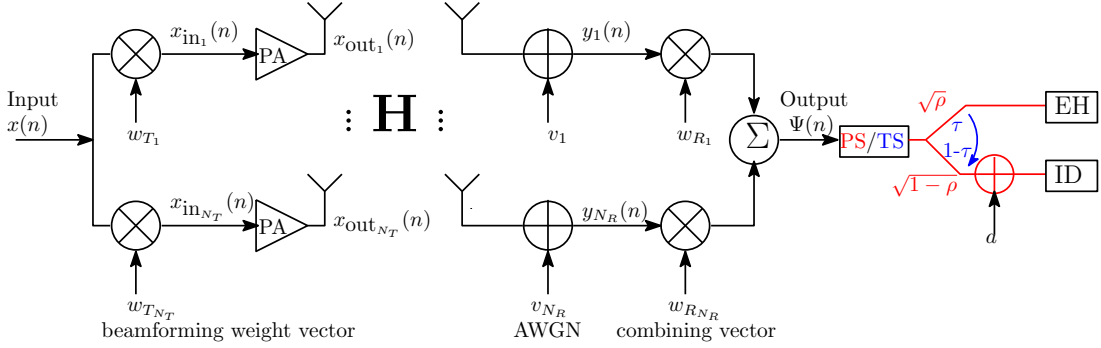}
\caption{Block diagram of the considered MIMO system with nonlinear dynamic HPAs.}
\label{fig:smodel}
\vspace{-4mm}
\end{figure*}
We consider a point-to-point MIMO system as illustrated in Fig. \ref{fig:smodel}, with $N_T$ transmit and $N_R$ receive antennas. Assuming a discrete time baseband channel with uncorrelated quasi-static frequency flat fading, the received signal in symbol time $n$ is
\vspace{-2mm}
\begin{equation}
\mathbf{y}=\mathbf{H}\mathbf{x_{\text{out}}}+\mathbf{v},
\end{equation}
where $\mathbf{y}$$=$$[y_1(n),\cdots,y_{N_R}(n)]^T$ is the $N_R\times1$ received signal vector, $\mathbf{x_{\text{out}}}$$=$$[x_{\text{out}_1}(n),\cdots,x_{\text{out}_{N_T}}(n)]^T$ is the $N_T\times1$ HPA output vector at the transmitter, $\mathbf{v}$ is the $N_R\times1$ noise vector with elements belonging to independent and identically distributed (i.i.d.) complex circular Gaussian distribution i.e., $\mathcal{CN}(0,\sigma_v^2)$ and uncorrelated with the transmitted symbols, and $\mathbf{H}$ denotes the $N_R\times N_T$ channel gain matrix. The entries of $\mathbf{H}$ are i.i.d. complex Gaussian random variables, each with a $\mathcal{CN}(0,1)$ distribution. The vector $\mathbf{x_{\text{out}}}$ can be expressed as $\mathbf{x_{\text{out}}}$$=$$f_{\text{HPA}}(\mathbf{x_{\text{in}}})$, where $f_{\text{HPA}}(\cdot)$ is the HPA function and $\mathbf{x_{\text{in}}}$$=$$[x_{\text{in}_1}(n),\cdots,x_{\text{in}_{N_T}}(n)]^T$ is the $N_T\times1$ HPA input vector. In particular, $\mathbf{x_{\text{in}}}$ depends on the $N_T\times1$ unit-norm beamforming vector $\mathbf{w}_{T}$$=$$[w_{T_1}(n),\cdots,w_{T_{N_T}}(n)]^T$ and the primary transmitted symbol $x(n)$ with power $P_{\text{t}}$.

Signals received from $N_R$ branches of the receiver are combined with a $N_R$$\times$$1$ combining vector $\mathbf{w}_R$$=$ $[w_{R_1}(n),\cdots,w_{R_{N_R}}(n)]^T$ to obtain the signal $\Psi(n)$ and fed to a SWIPT block, which can be both time-switching (TS) or power-switching (PS) in nature. Let $\rho,\tau \in [0,1]$ denote the power and time splitting parameter, respectively \cite{intro2}. The information decoder (ID) extracts the information content from $\Psi(n)$ in the presence of an additional circuit noise $a$, caused due to the RF-to-baseband conversion; this additional noise is modelled as $a \sim \mathcal{CN}(0,\sigma_{a}^2)$.

The harvested energy $\mathbb{E}_{\text{H}}$ can be expressed in terms of the input power $\xi$ at the energy harvester (EH) as $\mathbb{E}_{\text{H}}=p_h(\xi)$, where $p_h(\cdot)$ is a non-decreasing function of its argument.

In general, $p_h(\cdot)$ is a monotonic nonlinear function and is modelled by a piece-wise linear model \cite{plinear} i.e.,
\vspace{-2mm}
\begin{align} \label{hmodel}
\mathbb{E}_{\text{H}}=p_h(\xi)= 
\begin{cases}
0 & \xi \leq p_h^{l},\\
\eta\xi & p_h^{l} \leq \xi \leq p_h^{u},\\
\eta p_h^{u} & \xi > p_h^{u},
\end{cases}
\vspace{-2mm}
\end{align}
where $p_h^{l}$ is the minimum RF input power required for harvesting, $p_h^{u}$ is the saturation RF input power level above which the harvested output power remains practically constant, and $\eta \in [0,1]$ is the constant energy\footnote{Energy and power are used indistinctly in this work, which can be interpreted in terms of normalized symbol duration.} efficiency when $\xi \in [p_h^{l}, p_h^{u}]$.

\section{HPA model}

In this section, we characterize the HPA by taking into account the practical imperfections associated with it. This aspect of HPA modeling is extremely crucial for SWIPT, as it has a direct impact on the performance for both information and energy transfer (see the diagram in Fig. \ref{fig:smodel}). 

\subsection{HPA model characterization}

The practical HPA output differs from the ideal response due to the presence of nonlinearities and memory effects\cite{dpdtc} that severely degrade the system performance. A practical HPA output signal $x_{\text{out}}(n)$ depends on both the present and past HPA input signals i.e., $x_{\text{out}}(n)=f(x_{\text{in}}(n),x_{\text{in}}(n-1),x_{\text{in}}(n-2),\cdots),$ where $x_{\text{in}}(n)$ is the  HPA input signal during the $n-$th symbol duration. The memory effects are normally attributed to time delays and/or phase shifts in the matching networks and circuit elements used. Existing works \cite{pa1,pa2} study the impact of HPA nonlinearities by considering a memoryless model. As a result, the dynamic behavior of HPA is not captured. To closely mimic the nonlinear HPA behavior with memory effects, a low-complexity and mathematically tractable memory polynomial model (MPM) is adopted to model the HPA i.e.,
\vspace{-2mm}
\begin{equation} \label{def}
x_{\text{out}}(n)=\sum_{p=1}^{P}\sum_{m=0}^{M}c_{p,m}x_{\text{in}}(n-m)|x_{\text{in}}(n-m)|^{p-1},
\end{equation}
where $P$ is the non-linearity order, $M$ is memory depth, and $c_{p,m}$ with $p=1,\cdots,P,$ $m=0,\cdots,M,$ denotes the complex MPM coefficients \cite{dpdtc}. The memory depth $(M)$ implies that $x_{\text{out}}(n)$ depends on the past $M$ input samples i.e, $x_{\text{out}}(n)=f(x_{\text{in}}(n),\cdots,x_{\text{in}}(n-M))$. Accordingly, (\ref{def}) can be written in a generic form as $x_{\text{out}}(n)=\boldsymbol{\Phi}\text{C}_{\text{HPA}}$, where
\vspace{-2mm}
\begin{equation*}
\boldsymbol{\Phi}= 
\begin{bmatrix}
x_{\text{in}}(n) \\
\vdots \\
x_{\text{in}}(n-M)\\
x_{\text{in}}(n)|x_{\text{in}}(n)|\\
\vdots \\
x_{\text{in}}(n-M)|x_{\text{in}}(n-M)|\\
\vdots \\
x_{\text{in}}(n-M)|x_{\text{in}}(n-M)|^{P-1}
\end{bmatrix}^T, \quad \text{and}
\end{equation*}
$\text{C}_{\text{HPA}}=[c_{1,0},\cdots,c_{1,M},c_{2,0},\cdots,c_{2,M},\cdots,c_{P,M}]^T.$ Thus for a given set of input-output sample pairs, we recursively obtain $\text{C}_{\text{HPA}}$ by solving an over-determined system of equations by using a least square method \cite[Chapter~6]{pabk}.

\vspace{-2mm}

\subsection{Analytical insights into HPA output $x_{\text{out}}(n)$}

From (\ref{def}), we observe that $x_{\text{out}}(n)$ jointly depends on $x_{\text{in}}(n)$, HPA memory, and nonlinearities. By rewriting (\ref{def}), we have
\vspace{-2mm}
\begin{align} \label{buss1}
x_{\text{out}}(n)=&\sum_{p=1}^{P}\sum_{m=0}^{M}c_{p,m}x_{\text{in}}(n-m)|x_{\text{in}}(n-m)|^{p-1}\\
=& \underbrace{\sum_{p=1}^P c_{p,0}|x_{\text{in}}(n)|^{p-1}}_{\text{scaling factor}}\underbrace{x_{\text{in}}(n)}_{\text{desired signal}} \nonumber \\
& + \underbrace{\sum_{p=1}^{P}\sum_{m=1}^{M}c_{p,m}x_{\text{in}}(n-m)|x_{\text{in}}(n-m)|^{p-1}}_{\text{distortion}} \nonumber\\
=& \Delta_{\text{HPA}}(x_{\text{in}}(n))x_{\text{in}}(n) + \delta_{\text{HPA}}(x_{\text{in}}(n)), \nonumber
\end{align}
where $\Delta_{\text{HPA}}(x_{\text{in}}(n))$ and $\delta_{\text{HPA}}(x_{\text{in}}(n))$ denote the scaling factor and the distortion corresponding to the HPA input signal $x_{\text{in}}(n)$, respectively. From (\ref{buss1}), we observe that $\Delta_{\text{HPA}}(x_{\text{in}}(n))$ captures the nonlinearity of HPA, while $\delta_{\text{HPA}}(x_{\text{in}}(n))$ represents both the nonlinearities and memory effects.

Existing works \cite{pa1,pa2} characterize the impact of HPA nonlinearities based on Bussgang's theorem \cite{bgang}, which states that the HPA output signal $x_{\text{out}}(n)$ can be represented as a sum of the scaled useful signal and an uncorrelated nonlinear distortion component. However, we prove the contrary below.

\begin{prop}
The distortion generated at the output of each HPA is correlated with the input signal i.e.,
\vspace{-2mm}
\begin{equation}
\mathbb{E}\{x^*_{\text{in}}(n)\delta_{\text{HPA}}(x_{\text{in}}(n))\}\neq 0.
\end{equation}
\end{prop}
\begin{proof}
The proof is provided in the Appendix.
\end{proof}
As the proposition implies, the distortion is correlated with the input signal and this effect cannot be ignored, if we are investigating the role of HPA in SWIPT systems.

\section{HPA-based MRC/MRT Beamforming}

We study the impact of HPA on a point-to-point MIMO communication system that employs maximum-ratio combining/transmission (MRC/MRT) beamforming schemes. We consider MRC/MRT MIMO as a simple case to demonstrate the impact of HPA in MIMO SWIPT architectures.

From Section \ref{msec}, the HPA output vector $\mathbf{x_{\text{out}}}$$=$$f_{\text{HPA}}(\mathbf{x_{\text{in}}})$, where $f_{\text{HPA}}(\cdot)$ is the HPA function, and $\mathbf{x_{\text{in}}}$ is the HPA input vector. We also know that $\mathbf{x_{\text{in}}}$$=$$\mathbf{w}_Tx(n)$ depends on the primary transmitted symbol $x(n)$ and the beamforming vector $\mathbf{w}_T$. Hence, we have $\mathbf{x_{\text{out}}}$$=$$f_{\text{HPA}}(\mathbf{w}_Tx(n))$$=$$\mathbf{\Delta_{\text{HPA}}}\circ\mathbf{w}_Tx(n)+\boldsymbol{\delta_{\text{HPA}}}$, where $\circ$ refers to the Hadamard product, $\boldsymbol{\Delta}_{\text{HPA}}=[\Delta_{\text{HPA}_1}(x_{\text{in}_1}(n)),\cdots,\Delta_{\text{HPA}_{N_T}}(x_{\text{in}_{N_T}}(n))]^T$, and $\boldsymbol{\delta}_{\text{HPA}}=[\delta_{\text{HPA}_1}(x_{\text{in}_1}(n)),\cdots,\delta_{\text{HPA}_{N_T}}(x_{\text{in}_{N_T}}(n))]^T$. Accordingly, we obtain
\vspace{-2mm}
\begin{align}
\mathbf{y}=&\:\:\mathbf{H}(\mathbf{\Delta}_\text{HPA}\circ \mathbf{w}_T x(n)+\boldsymbol{\delta_{\text{HPA}}})+\mathbf{v} \\
=&\:\:\mathbf{H}(\mathbf{\Delta}_\text{HPA}\circ \mathbf{w}_T) x(n)+\mathbf{H}\boldsymbol{\delta_{\text{HPA}}}+\mathbf{v}. \nonumber
\end{align}

At the receiver, signals from all the branches are weighted by the combining vector $\mathbf{w}_R$ to obtain the detected signal as
\vspace{-2mm}
\begin{align} \label{dvar}
\Psi(n)=&\mathbf{w}_R^H\mathbf{y} \\
=&\mathbf{w}_R^H\mathbf{H}(\mathbf{\Delta}_\text{HPA}\circ \mathbf{w}_T) x(n)+\mathbf{w}_R^H\mathbf{H}\boldsymbol{\delta_{\text{HPA}}}+\mathbf{w}_R^H\mathbf{v}, \nonumber
\end{align}
where $(\cdot)^H$ represents the Hermitian transpose of a matrix. To maximize the signal component, the optimal $\mathbf{w}_T$ is chosen as
\vspace{-2mm}
\begin{equation} \label{wt}
\mathbf{w}_T=\frac{\mathbf{H}^H\mathbf{w}_R}{||\mathbf{H}^H\mathbf{w}_R||} \varoslash \mathbf{\Delta}_\text{HPA},
\end{equation}
where $\varoslash$ denotes the Hadamard division and $||\cdot||$ is the Euclidean norm. Substituting (\ref{wt}) in (\ref{dvar}) we obtain
\vspace{-2mm}
\begin{equation}
\Psi(n)=||\mathbf{H}^H\mathbf{w}_R||x(n)+\mathbf{w}_R^H\mathbf{H}\boldsymbol{\delta_{\text{HPA}}}+\mathbf{w}_R^H\mathbf{v}.
\end{equation}
It is worth noting that due to the nonlinearity and the memory effects of HPA, we have the additional term $\mathbf{w}_R^H\mathbf{H}\boldsymbol{\delta_{\text{HPA}}}$. As a result, the output SNR conditioned on $\mathbf{H}$ is
\vspace{-2mm}
\begin{align}
\gamma&=\frac{P_{\text{t}}||\mathbf{H}^H\mathbf{w}_R||^2}{||\mathbf{w}_R^H\mathbf{H}\boldsymbol{\delta_{\text{HPA}}}||^2+||\mathbf{w}_R^H\mathbf{v}||^2}\\
&=\frac{P_{\text{t}}\mathbf{w}_R^H (\mathbf{H}\mathbf{H}^H)\mathbf{w}_R}{||\mathbf{w}_R^H\mathbf{H}\boldsymbol{\delta_{\text{HPA}}}||^2+\sigma_v^2||\mathbf{w}_R||^2}. \nonumber
\end{align}
\vspace{-2mm}
\begin{figure}[t]
\centering\includegraphics[width=0.56\linewidth]{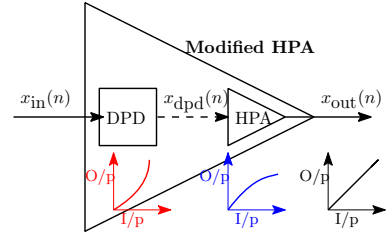}
\caption{Predistortion system (Cascade of DPD followed by HPA).}
\label{fig:model}
\vspace{-4mm}
\end{figure}

\begin{figure*}[!t]
\begin{align} \label{reregion}
\mathbb{C}_{RE}= 
\begin{cases}
\displaystyle \bigcup\limits_{0 \leq \tau \leq 1} \left\lbrace (R,E): R \leq (1-\tau)\log_2\left(1+ \frac{P_{\text{t}}\Lambda_{\max}||\mathbf{w}_R^*||^2}{\zeta||\mathbf{w}_R^{*H}\mathbf{H}\boldsymbol{\delta_{\text{HPA}}}||^2+\sigma_v^2||\mathbf{w}_R^*||^2+\sigma_a^2} \right), E \leq \tau p_h(\Psi ) \right\rbrace , \quad \:\:\text{TS}\\
\displaystyle \bigcup\limits_{0 \leq \rho \leq 1} \left\lbrace (R,E): R \leq \log_2\left(1+ \frac{(1-\rho)P_{\text{t}}\Lambda_{\max}||\mathbf{w}_R^*||^2}{(1-\rho)(\zeta||\mathbf{w}_R^{*H}\mathbf{H}\boldsymbol{\delta_{\text{HPA}}}||^2+\sigma_v^2||\mathbf{w}_R^*||^2)+\sigma_a^2} \right), E \leq \rho p_h(\Psi) \right\rbrace , \quad \text{PS}.
\end{cases}
\end{align}
\hrulefill
\vspace{-2mm}
\end{figure*}

Note that if we ignore the HPA effects, we have $\boldsymbol{\delta_{\text{HPA}}}=0$. Hence, maximizing $\gamma$ with respect to $\mathbf{w}_R$ is equivalent to finding the squared-spectral norm of $\mathbf{H}.$ This squared-spectral norm is the maximum eigen-value of $\mathbf{H}\mathbf{H}^H$ i.e., $\Lambda_{\max}$, and the eigenvector of $\mathbf{H}\mathbf{H}^H$ corresponding to $\Lambda_{\max}$ is the optimum $\mathbf{w}_R^*,$ which maximizes $\gamma$ to yield $\gamma_{\max}$ \cite{mallik}. As a result when $\boldsymbol{\delta_{\text{HPA}}}=0$, we obtain the following classical expression of $\gamma_{\max}$ corresponding to an ideal $N_T$$\times$$N_R$ MIMO system i.e., $\displaystyle \gamma_{\max}=\frac{P_{\text{t}}}{\sigma_v^2}\Lambda_{\max}$; however this is not the case here as $\boldsymbol{\delta_{\text{HPA}}}\neq0$.


We derive the RE region of both PS and TS architectures in (\ref{reregion}) by using $\gamma_{\max}$. The effect of HPA nonlinearities and memory effects can be seen in (\ref{reregion}), where we note the presence of an additional term i.e., $||\mathbf{w}_R^{*H}\mathbf{H}\boldsymbol{\delta_{\text{HPA}}}||^2$ in the denominator; $\zeta$$\in$$\{0,1\}$ is a binary variable denoting the effect of HPA on $\gamma_{\max}$. We have $\zeta=0$ when predistortion is employed (discussed next) and $\zeta=1$ otherwise.

\section{HPA Predistortion Scheme}

In the previous section, we observe that $\gamma_{\max}$ is affected by the HPA non-linearities and memory effects, which in turn negatively affects the RE region (in both PS and TS architectures). Higher the signal power, greater is the value of $||\mathbf{w}_R^{*H}\mathbf{H}\boldsymbol{\delta_{\text{HPA}}}||^2$ and more significant is its effect on $\gamma_{\max}$ in (\ref{reregion}). To overcome this effect, we propose the use of digital predistortion (DPD) \cite{rev1} technique.

To linearly amplify a high PAPR signal, the peak power of the HPA input signal must remain in the linear region of the HPA and this affects power efficiency, which is a crucial parameter for SWIPT. The DPD scheme allows a higher operating point since enables the HPA to operate linearly over its full range up to saturation. The addition of DPD to HPA, is analogous to precoding in conventional wireless communication. Hence, the HPA block in Fig. \ref{fig:smodel} is essentially now the joint DPD-HPA block as seen in Fig. \ref{fig:model}.

An ideal DPD function is the inverse of the HPA model \cite{dpdtc} derived in (\ref{buss1}). However, it is difficult to analytically characterize a DPD function \footnote{We observe that the number of coefficients $P(M+1)$ grows rapidly with $P$ and/or $M$. However as it can be seen from the numerical results, small values of $P,M$ are sufficient to  produce an accurate HPA model. Please note that the additional cost of using DPD at the transmitter end, along with its associated complexity is beyond the scope of this work.}. Hence, we follow an iterative approach and propose Algorithm \ref{Algo:AL1} for obtaining an estimate of the optimal DPD function $x_{\text{dpd}}^{\text{opt}}(n)$.

From (\ref{buss1}), we can express $x_{\text{out}}(n)$ as
\vspace{-2mm}
\begin{equation}
x_{\text{out}}(n)=\Delta_{\text{HPA}}(x_{\text{dpd}}(n))x_{\text{dpd}}(n) + \delta_{\text{HPA}}(x_{\text{dpd}}(n)),
\end{equation}
where $x_{\text{dpd}}(n)$ is the DPD output signal corresponding to $x_{\text{in}}(n).$ Ideally, we expect the HPA output to be identical to the DPD input i.e., $x_{\text{out}}(n)=x_{\text{in}}(n)$. Hence we obtain
\vspace{-2mm}
\begin{equation}
x_{\text{dpd}}(n)=\frac{1}{\Delta_{\text{HPA}}(x_{\text{dpd}}(n))}(x_{\text{in}}(n)-\delta_{\text{HPA}}(x_{\text{dpd}}(n))),
\end{equation}
by assuming known HPA characteristics i.e., $M,P,c_{pm}$ $\forall$ $m=1,\cdots,M$ and $p=1,\cdots,P$. However, we observe from (\ref{buss1}) that $\Delta_{\text{HPA}}(x_{\text{dpd}}(n))$ is a function of $|x_{\text{dpd}}(n)|$ as well. Algorithm \ref{Algo:AL1} converges to $x_{\text{dpd}}^{\text{opt}}(n)$ after a few iterations due to the adopted least square method in (\ref{def}). In this way we obtain the DPD function in an iterative manner.
\vspace{-2mm}
\begin{algorithm}[t]
{\small
\caption{Algorithm to estimate $x_{\text{dpd}}^{\text{opt}}(n)$\label{Algo:AL1}}
\begin{algorithmic}[1]
\Require $x_{\text{in}}(n),M,P,c_{pm}$ $\forall$ $m=1,\cdots,M$ and $p=1,\cdots,P$, and acceptable tolerance $\epsilon>0$
\Ensure $x_{\text{dpd}}^{\text{opt}}(n)$
\State Initialize $x_{\text{dpd}}(n)$ with $x_{\text{in}}(n)$
\State Set $k=0$
\State Calculate estimate $\widehat{x}_{\text{dpd}\_k}(n)$ using $\Delta_{\text{HPA}}(x_{\text{dpd}}(n))$
\While{$|\widehat{x}_{\text{dpd}\_k}(n)-x_{\text{dpd}}(n)|>\epsilon$}
\State Calculate a better $\widehat{x}_{\text{dpd}\_k}(n)$ using $|\widehat{x}_{\text{dpd}\_k}(n)|$ in $\Delta_{\text{HPA}}(\cdot)$
\State Set $k=k+1$
\EndWhile
\State Set $x_{\text{dpd}}^{\text{opt}}(n)=\widehat{x}_{\text{dpd}\_k}(n)$
\end{algorithmic}
}
\end{algorithm}

\section{Numerical Results}
Computer simulations are carried-out to evaluate the effects of HPA on SWIPT performance as well as the efficiency of the proposed DPD scheme. 

\subsection{Validation of Proposed Model}

Fig. \ref{fig:psd} shows the power spectra of a 20-MHz  four-channel WCDMA signal at the output of HPA, both with and without DPD. Non-linearity order $P=7$ and memory depth $M=3$ have been considered. These values were selected because the MPM attained a minimum error with this set of values.

Simulated outputs confirm the spectrum regrowth caused by the HPA nonlinearity and memory effects. The figure also demonstrates that the out-of-band distortion has been reduced considerably due to DPD and that the in-band distortion has been eliminated as well. This is evident in the improvement in adjacent channel power ratio values from $25$ dBc without DPD to $47$ dBc after DPD linearization.
\begin{figure}[t]
\centering\includegraphics[width=0.92\linewidth]{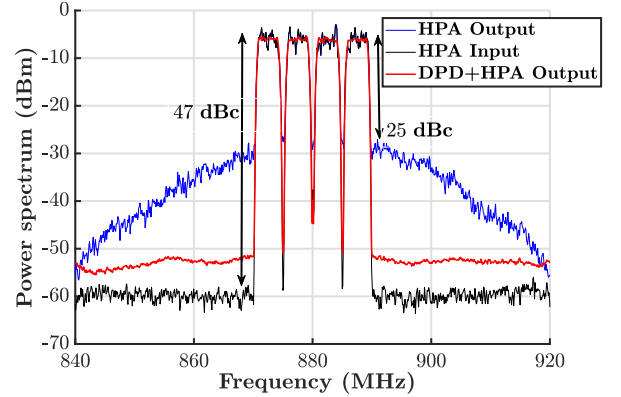}
\caption{Simulated output spectra of class-AB HPA \cite{power} before and after the DPD linearization.}
\vspace{-4mm}
\label{fig:psd}
\end{figure}

Fig. \ref{fig:ccdf} demonstrates the complementary cumulative distribution of the output signal against its corresponding PAPR. Note that the output power is significantly clipped in the absence of a DPD; a performance gap of approximately $3$ dB can be observed. As a result, a low power is transmitted regardless of the actual power of signal, which leads to low power reaching the EH block as depicted in Fig. \ref{fig:smodel}. On the other hand, we observe that the PAPR of the HPA output signal with DPD is very close to that of the HPA input signal i.e., DPD improves the PAPR and refines the original signal by increasing the PAPR of the predistorted signal. Hence, we can state that predistortion is extremely beneficial to applications like SWIPT, where high PAPR signals are extremely useful.
\begin{figure}[t]
\centering\includegraphics[width=0.92\linewidth]{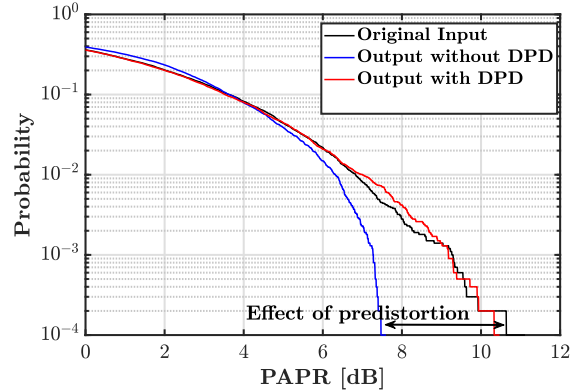}
\caption{Effect of predistortion on the HPA output signal.}
\vspace{-4mm}
\label{fig:ccdf}
\end{figure}

\subsection{Effect on SWIPT}

We investigate the effect of predistortion on SWIPT. We consider a MIMO system with $N_T=3$ and $N_R=2$ in a Rayleigh fading scenario with $12$ metres Tx-Rx distance, $2.6$ pathloss exponent, and $14$ dBm transmit power. Values of $p_h^l,p_h^u,$ and $\eta$ in (\ref{hmodel}) are $-10$ dBm, $2$ dBm, and $24\%$ respectively, and $\sigma_v^2=-70$ dBm and $\sigma_a^2=-50$ dBm. We consider a class AB HPA \cite{power} and assume that the HPAs in different antennas exhibit identical input-output relation \cite{pa2}.

\subsubsection{Spectral Efficiency}

Fig. \ref{fig:rate} demonstrates the effect of DPD on the spectral efficiency of considered system, where marginal improvement of $\sim 3\%$ is observed. Moreover the performance of the schemes, both with and without DPD, merge as $P_{\text{t}}$ enters the saturation region; this is intuitive as DPD enables the HPA to linearize its performance upto the saturation point and not beyond.

\begin{figure}[t]
\centering\includegraphics[width=0.94\linewidth]{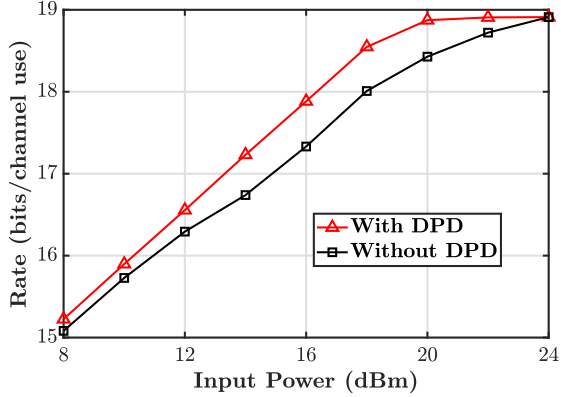}
\caption{Effect of HPA distortion and predistortion on data rate.}
\label{fig:rate}
\vspace{-4mm}
\end{figure}
\subsubsection{RE region}

We investigate the predistortion effect on SWIPT in terms of RE region enhancement. We observe from (\ref{reregion}) that the RE region without DPD suffers from an additional distortion, which limits the system performance. However, this distortion is eliminated by the use of the DPD block. Fig. \ref{fig:region} demonstrates that the RE region is enhanced by approximately $24\%$ in both the PS and TS scenarios due to DPD. Note that the maximum performance enhancement is observed at the points A and B $(\sim 21.04\%)$, while the minimum at the points C and D $(\sim 3.36\%)$. This is because A, B refer to only energy harvesting with no information decoding while C, D imply the other extreme. This means that a predistortion does not result in much gain in terms of data rate (as seen in Fig. \ref{fig:rate}), as in terms of harvested energy. This interesting observation demonstrates that the performance enhancement will be more significant if predistortion is used in applications, where energy transfer is the primary objective i.e., $\rho,\tau \rightarrow 1$ in PS and TS architectures, respectively.

\section{Conclusion}
In this work, we have studied the effects of HPA on the performance of a point-to-point MIMO SWIPT system. We have shown that dynamic HPA is characterized by nonlinearities and memory that significantly degrade the achievable RE region. It has been shown that the employment of a digital predistorter before the HPA at the transmitter, enables the transmission of unclipped high PAPR signals. The predistortion results in a RE region enhancement of the considered system; approximately $24\%$ gain for both PS and TS architectures. A promising extension of this work is to consider more sophisticated MIMO SWIPT architectures.

\begin{appendix}
The correlation between the input signal and the distortion
generated at the output of HPA
\begin{align}
&\mathbb{E}\{x^*_{\text{in}}(n)\delta_{\text{HPA}}(x_{\text{in}}(n))\} \\
&= \mathbb{E}\left\lbrace x^*_{\text{in}}(n)\sum_{p=1}^{P}\sum_{m=1}^{M}c_{p,m}x_{\text{in}}(n-m)|x_{\text{in}}(n-m)|^{p-1} \right\rbrace \nonumber \\
&= \sum_{p=1}^{P}\sum_{m=1}^{M}c_{p,m}\mathbb{E}\{x^*_{\text{in}}(n)x_{\text{in}}(n-m)\}\mathbb{E}\{|x_{\text{in}}(n-m)|^{p-1}\} \nonumber \\
&\overset{(a)}{=} \sum_{p=1}^{P}\sum_{m=1}^{M}c_{p,m}R_{m,n}\mathbb{E}\{|x_{\text{in}}(n-m)|^{p-1}\} \neq  0, \nonumber
\end{align}
where $R_{m,n}$ in (a) is the auto-correlation function corresponding to $x_{\text{in}}(n)$ and $x_{\text{in}}(n-m)$. If $x_{\text{in}}$ is wide sense stationary in nature, then we have $R_{m,n}=R_{m}$ $\forall$ $n.$ Hence the proof.
\end{appendix}

\begin{figure}[t]
\centering\includegraphics[width=0.94\linewidth]{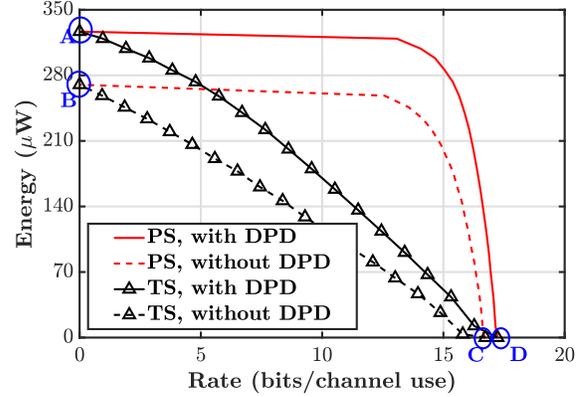}
\caption{Achievable rate-energy region with and without predistortion.}
\vspace{-5mm}
\label{fig:region}
\end{figure}

\bibliographystyle{IEEEtran}
\bibliography{krikidis_WCL0842_refs}

\begin{thebibliography}{10}
\providecommand{\url}[1]{#1}
\csname url@samestyle\endcsname
\providecommand{\newblock}{\relax}
\providecommand{\bibinfo}[2]{#2}
\providecommand{\BIBentrySTDinterwordspacing}{\spaceskip=0pt\relax}
\providecommand{\BIBentryALTinterwordstretchfactor}{4}
\providecommand{\BIBentryALTinterwordspacing}{\spaceskip=\fontdimen2\font plus
\BIBentryALTinterwordstretchfactor\fontdimen3\font minus
  \fontdimen4\font\relax}
\providecommand{\BIBforeignlanguage}[2]{{%
\expandafter\ifx\csname l@#1\endcsname\relax
\typeout{** WARNING: IEEEtran.bst: No hyphenation pattern has been}%
\typeout{** loaded for the language `#1'. Using the pattern for}%
\typeout{** the default language instead.}%
\else
\language=\csname l@#1\endcsname
\fi
#2}}
\providecommand{\BIBdecl}{\relax}
\BIBdecl

\bibitem{intro1}
I.~{Krikidis}, S.~{Timotheou}, S.~{Nikolaou}, G.~{Zheng}, D.~W.~K. {Ng}, and
  R.~{Schober}, ``Simultaneous wireless information and power transfer in
  modern communication systems,'' \emph{IEEE Commun. Mag.}, vol.~52, no.~11,
  pp. 104--110, Nov. 2014.

\bibitem{intro2}
B.~{Clerckx}, R.~{Zhang}, R.~{Schober}, D.~W.~K. {Ng}, D.~I. {Kim}, and H.~V.
  {Poor}, ``Fundamentals of wireless information and power transfer: From {RF}
  energy harvester models to signal and system designs,'' \emph{IEEE J. Selec.
  Areas Commun.}, vol.~37, no.~1, pp. 4--33, Jan. 2019.

\bibitem{srv}
N.~{Xia}, H.~{Chen}, and C.~{Yang}, ``Radio resource management in
  machine-to-machine communications—{A} survey,'' \emph{IEEE Commun. Surveys
  Tuts.}, vol.~20, no.~1, pp. 791--828, Firstquarter 2018.

\bibitem{plinear}
O.~L.~A. {López}, H.~{Alves}, R.~D. {Souza}, and S.~{Montejo-Sánchez},
  ``Statistical analysis of multiple antenna strategies for wireless energy
  transfer,'' \emph{IEEE Trans. Commun.}, vol.~67, no.~10, pp. 7245--7262, Oct.
  2019.

\bibitem{papr1}
A.~{Collado} and A.~{Georgiadis}, ``Improving wireless power transmission
  efficiency using chaotic waveforms,'' in \emph{Proc. IEEE MTT-S Int. Dig.},
  Montreal, QC, Canada, June 2012, pp. 1--3.

\bibitem{papr2}
D.~I. {Kim}, J.~H. {Moon}, and J.~J. {Park}, ``New {SWIPT} using {PAPR}: How it
  works,'' \emph{IEEE Wireless Commun. Lett.}, vol.~5, no.~6, pp. 672--675,
  Dec. 2016.

\bibitem{mmry}
P.~{Roblin}, D.~E. {Root}, J.~{Verspecht}, Y.~{Ko}, and J.~P. {Teyssier}, ``New
  trends for the nonlinear measurement and modeling of high-power {RF}
  transistors and amplifiers with memory effects,'' \emph{IEEE Trans. Microw.
  Theory Techn.}, vol.~60, no.~6, pp. 1964--1978, June 2012.

\bibitem{dpdtc}
X.~{Yu} and H.~{Jiang}, ``Digital predistortion using adaptive basis
  functions,'' \emph{IEEE Trans. Circuits Syst. I, Reg. Papers}, vol.~60,
  no.~12, pp. 3317--3327, Dec. 2013.

\bibitem{pa1}
J.~{Qi} and S.~{Aissa}, ``On the power amplifier nonlinearity in {MIMO}
  transmit beamforming systems,'' \emph{IEEE Trans. Commun.}, vol.~60, no.~3,
  pp. 876--887, Mar. 2012.

\bibitem{pa2}
N.~N. {Moghadam}, G.~{Fodor}, M.~{Bengtsson}, and D.~J. {Love}, ``On the energy
  efficiency of {MIMO} hybrid beamforming for millimeter-wave systems with
  nonlinear power amplifiers,'' \emph{IEEE Trans. Wireless Commun.}, vol.~17,
  no.~11, pp. 7208--7221, Nov. 2018.

\bibitem{pabk}
F.~M. Ghannouchi, O.~Hammi, and M.~Helaoui, \emph{Behavioral modeling and
  predistortion of wideband wireless transmitters}.\hskip 1em plus 0.5em minus
  0.4em\relax John Wiley \& Sons Ltd, 2015.

\bibitem{bgang}
J. J. Bussgang, ``Crosscorrelation functions of amplitude-distorted Gaussian
  signals,'' Res. Lab. Elec., Massachusetts Inst. Technol., Cambridge, MA, USA,
  Tech. Rep. 216, Mar. 1952.

\bibitem{mallik}
P.~A. {Dighe}, R.~K. {Mallik}, and S.~S. {Jamuar}, ``Analysis of
  transmit-receive diversity in {Rayleigh} fading,'' \emph{IEEE Trans.
  Commun.}, vol.~51, no.~4, pp. 694--703, Apr. 2003.

\bibitem{rev1}
A.~{Katz}, J.~{Wood}, and D.~{Chokola}, ``The evolution of {PA} linearization:
  From classic feedforward and feedback through analog and digital
  predistortion,'' \emph{IEEE Microw. Mag.}, vol.~17, no.~2, pp. 32--40, 2016.

\bibitem{power}
M.~{Faulkner} and T.~{Mattsson}, ``Spectral sensitivity of power amplifiers to
  quadrature modulator misalignment,'' \emph{IEEE Trans. Veh. Technol.},
  vol.~41, no.~4, pp. 516--525, Nov. 1992.

\end{thebibliography}

\end{document}